\documentclass[conference]{IEEEtran}
\IEEEoverridecommandlockouts

\usepackage{cite}

\usepackage{amsmath,amssymb,amsfonts}
\usepackage{algorithmic}
\usepackage{graphicx}
\usepackage{textcomp}
\usepackage{xcolor}
\usepackage{amsmath} 
\usepackage{amsthm}
\usepackage{amsfonts} 
\usepackage{bm} 

\DeclareMathOperator{\sinc}{sinc}

\newcommand\numberthis{\addtocounter{equation}{1}\tag{\theequation}}

\newtheorem{theorem}{Theorem}

\def\BibTeX{{\rm B\kern-.05em{\sc i\kern-.025em b}\kern-.08em
    T\kern-.1667em\lower.7ex\hbox{E}\kern-.125emX}}
\begin{document}

\title{Maximum Eigenvalue Detection based Spectrum Sensing in RIS-aided System with Correlated Fading\vspace{-1.1ex}}

\author{\IEEEauthorblockN{Nikhilsingh Parihar}
\IEEEauthorblockA{
\textit{IIIT-Hyderabad}\\
Hyderabad, India \\
nikhilsingh.parihar@research.iiit.ac.in}
\vspace*{-1.5cm}
\and
\IEEEauthorblockN{Praful D. Mankar}
\IEEEauthorblockA{
\textit{IIIT-Hyderabad}\\
Hyderabad, India \\
praful.mankar@iiit.ac.in}
\vspace*{-1.5cm}
\and

\IEEEauthorblockN{Sachin Chaudhari}
\IEEEauthorblockA{
\textit{IIIT-Hyderabad}\\
Hyderabad, India \\
sachin.c@iiit.ac.in}\\
\vspace*{-1.5cm}

}
\maketitle

\begin{abstract}
Robust spectrum sensing is crucial for facilitating opportunistic spectrum utilization for secondary users (SU) in the absence of primary users (PU). However, propagation environment factors such as multi-path fading, shadowing, and lack of line of sight (LoS) often adversely affect detection performance. To deal with these issues, this paper focuses on utilizing reconfigurable intelligent surfaces (RIS) to improve spectrum sensing in the scenario wherein both the multi-path fading and noise are correlated. In particular, to leverage the spatially correlated fading, we propose to use maximum eigenvalue detection (MED) for spectrum sensing. We first derive exact distributions of test statistics, i.e., the largest eigenvalue of the sample covariance matrix, observed under the null and signal present hypothesis. Next, utilizing these results, we present the exact closed-form expressions for the false alarm and detection probabilities. In addition, we also optimally configure the phase shift matrix of RIS such that the mean of the test statistics is maximized, thus improving the detection performance. Our numerical analysis demonstrates that the MED's receiving operating characteristic (ROC) curve improves with increased RIS elements, SNR, and the utilization of statistically optimal configured RIS.
\end{abstract}

\vspace*{-2.5ex}

\section{Introduction}
\vspace*{-0.5ex}
Reconfigurable Intelligent Surfaces (RIS) are emerging as a viable solution for next-generation wireless communication systems due to their ability to provide reliable links, offering several advantages like extended coverage, high data rates, and enhanced sensing abilities. RIS comprises a low-cost uniform planar array (UPA) with numerous sub-wavelength-sized passive metamaterial elements. These elements can adjust their physical properties dynamically, allowing them to induce phase shifts in reflected signals. Such property enables RIS to control the wireless propagation environment between transmitter and receiver partially \cite{Qingqing_TowardsSmart,Basar,Emil_SigPro,Ozdogan}. This makes RIS attractive for reliable communication and for enabling efficient sensing and localization \cite{Emil_SigPro,Renzo_SmartRadio}. This paper aims to utilize the RIS to design a robust spectrum sensing technique to facilitate dynamic spectrum management in cognitive radio (CR) networks.
\vspace*{-1.2ex}

\subsection{Related Works}
\vspace*{-1.0ex}

Over the past two decades, significant research endeavors have been dedicated to designing efficient spectrum sensing techniques in the literature for conventional wireless communication systems. These techniques encompass energy detection (ED), maximum eigenvalue detection (MED), matched filtering (MF), and cyclostationary detection (CD).
However, in such conventional communication setups, achieving a high detection probability often necessitates a large number of observation samples, especially when the number of antennas at both the transmitter and receiver is limited and the wireless fading is severe. Consequently, this limitation restricts the opportunistic utilization of spectrum as a significant portion of the time is spent for sensing the spectrum \cite{Liang2008SensingSharingTradeoff}. 
This problem can be solved by using an RIS for spectrum sensing because of its ability to control the fading environment by providing a reliable indirect link.

In this direction, a few papers study the benefits of using RIS for spectrum sensing. 
For instance, \cite{Wu2021EnergyDetection} investigates RIS-enhanced ED for spectrum sensing. Therein, the authors employ the Gamma distribution approximation and central limit theorem to derive the closed-form expressions for the probability of detection. The authors of  \cite{Xie2023EnhancingPassiveorActiveSensing} study the impact of active and passive RIS-aided ED for spectrum sensing. The authors also analyze the
number of configurations and how many reflecting elements are needed for active and passive RIS to achieve
a detection probability close to 1. 
In \cite{Shaoe2022}, a RIS-aided weighted ED (WED) is studied for spectrum sensing wherein the IRS reflection is considered to dynamically change over time according to a designed codebook to substantially vary the received signal at the SU. The analyses presented in \cite{Xie2023EnhancingPassiveorActiveSensing} and \cite{Shaoe2022} heavily rely on the Gaussian approximation of test statistics, assuming a large number of samples are available for sensing. 

On the other hand, the authors of \cite{Jungang2022} employ MED scheme for spectrum sensing using RIS under correlated fading, wherein the largest eigenvalue of the sample covariance matrix of observed samples is used as the test statistics. 
Similarly, the authors of \cite{Sharma2014MEDCorrelatedNoise_TracyWidom} apply MED for spectrum sensing under correlated noise. 
MED performs better than ED  when the observed signal exhibits correlation  \cite{Zeng2008MED}.
However, applying MED leads to deriving the distribution of the largest eigenvalue of the sample covariance matrix that depends on the system under consideration. 
The authors of \cite{Jungang2022,Sharma2014MEDCorrelatedNoise_TracyWidom} show the resemblance of the sample covariance matrix with the single spiked model from random matrix theory and invoke the application of asymptotic distribution of largest eigenvalue from \cite{ge2021large} to obtain the detection and false alarm probabilities in the asymptotic region where the number of antennas at SU and the number of observation approaches to infinity. However, assuming the large number of antennas at the SU is impractical, and the large number of observation samples defeats the purpose of using RIS for efficient sensing, as discussed above. Moreover, as mentioned earlier, the analyses mainly rely on the large dimensional system, allowing for the convenient application of the law of large numbers or asymptotic results to make the analysis tractable. However, these analyses may merely serve as approximations for practical systems in which the number of antennas and observations is small or moderate. Motivated by this, our focus in this paper is to provide exact performance characterization of MED for RIS-aided spectrum sensing. 
\vspace*{-1.3ex}

\subsection{Contributions}
\vspace{-1.0ex}

This paper aims to design an efficient spectrum sensing for RIS-aided CR networks. In particular, we consider sensing of a single antenna PU 
at the SU equipped with multiple antennas and assisted by RIS. The SU is assumed to be capable of configuring RIS phase shifts to improve the test statistics for the detection. The main contributions of this paper are listed below.
\vspace{-1.0ex}

\begin{enumerate}
    \item This paper proposes MED  for spectrum sensing for RIS-aided CR networks in the presence of correlated fading and correlated noise.
\vspace{-0.5ex}
    
    \item The exact distribution of the largest eigenvalue of the central Wishart matrix is derived, which is then used to obtain the proposed MED's detection and false alarm probabilities.
\vspace{-0.5ex}
    \item We also present a near-optimal approach to configure RIS phase shift that maximizes the expected largest eigenvalue of the received signal sample covariance matrix.
\vspace{-0.5ex}
    \item Our numerical analysis demonstrates that increasing the number of RIS elements and configuring RIS phase shifts optimally will enhance the ROC. We also demonstrate that the detection probability of MED approaches $1$ for a moderate number of RIS elements, especially with optimally configured RIS. 
\end{enumerate}
\vspace{-0.7ex}

\section{System Model}
\vspace{-0.7ex}

This paper considers a wireless communication system comprising a single-antenna PU,  an SU  with $M$ antennas, and an RIS with $N$ elements, as shown in Fig. \ref{fig:system_model}. Considering that the RIS is deployed to enhance the spectrum sensing capability, it is reasonable to assume that the RIS is strategically placed such that it provides a strong LoS link to SU. Hence, we can model the channel array response of the RIS-SU link as  \cite{WangPeilan2020Channelarrayresponse}
\begingroup
\setlength\abovedisplayskip{0.3ex}
\setlength\belowdisplayskip{0.3ex}
\begin{align*}
    \mathbf{H}= \mathbf{a}\left(\theta, M\right)\left[\mathbf{a}\left(\cos{\phi}, N_x\right) \otimes \mathbf{a}\left(\sin{\phi}\cos{\varphi}, N_y\right)\right]^T,
\end{align*}
\endgroup
\begingroup
\setlength{\belowcaptionskip}{-10pt}
\begin{figure}[t]
  \centering\vspace{-3mm}
    \includegraphics[width=0.35\textwidth]{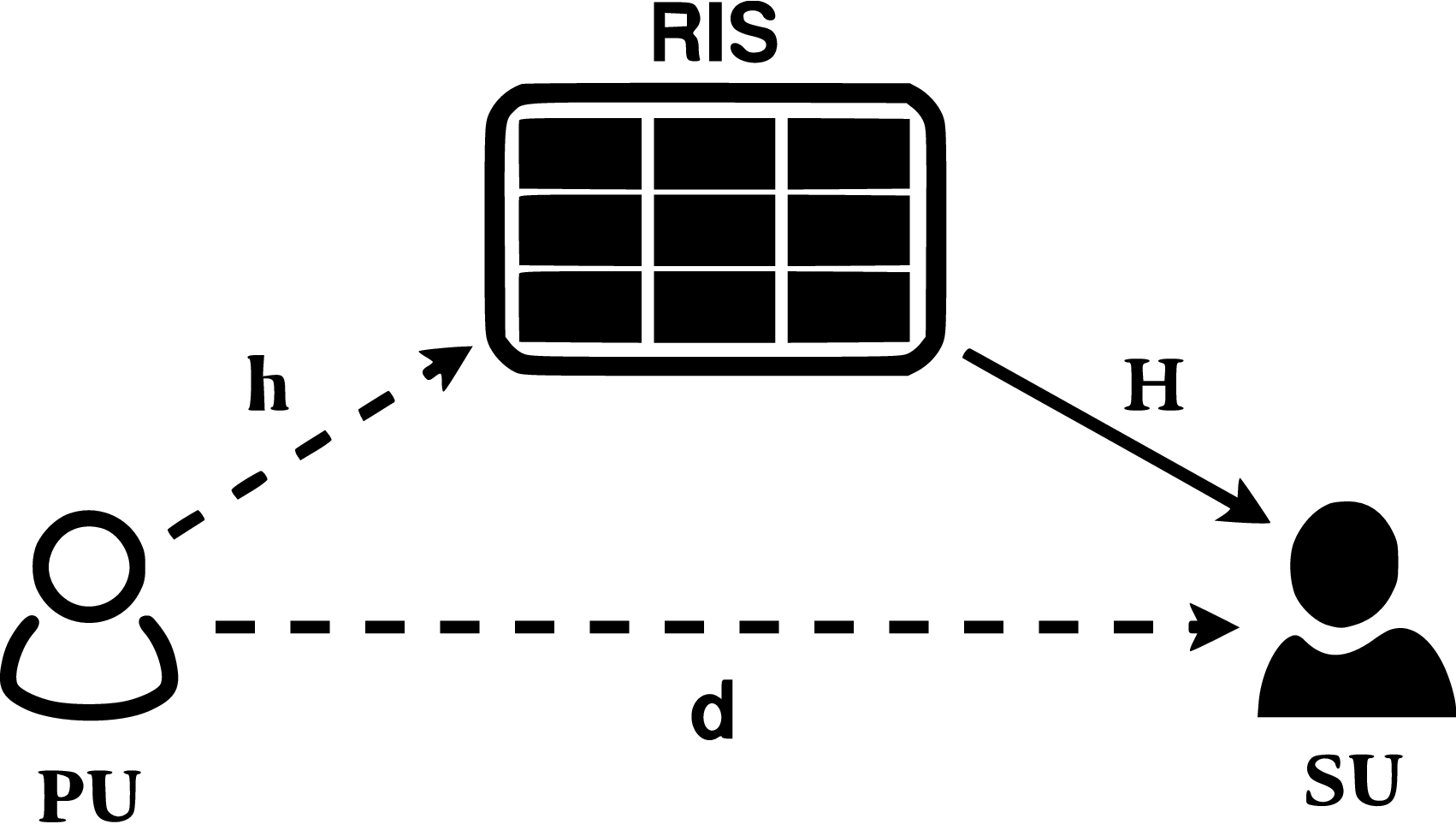}\vspace{-3mm}
  \caption{Illusration of the RIS-aided spectrum sensing system involving single antenna PU, a RIS equipped with N elements and M antenna elements at SU.}
  \vspace{-3.3ex}
  \label{fig:system_model}
\end{figure}%
\endgroup
where $\mathbf{a}\left(x,q\right)  = q^{-\frac{1}{2}}\left[1,~e^{jx}, \dots,~e^{j(q-1)x} \right]^T$
is a steering vector, $\otimes$ is the Kronecker product, $ N_x$ and $N_y$ are the number of RIS elements placed horizontally and vertically, respectively, such that $N_x \times N_y = N$. The $\phi$ and $\varphi$ are the azimuth and elevation angles associated with the RIS departure link, and $\theta$ is the angle of arrival at SU.
Further, we consider that PU-SU and PU-RIS links are independent and do not have a LoS component. To capture a generalized scenario, we assume that 1) a spatial correlation exists between the RIS/SU antenna elements and 2) the noise components at SU antennas are correlated. Thus, we model the PU-SU channel $\mathbf{d}$ and PU-RIS channel $\mathbf{h}$ using correlated Rayleigh fading such that $\mathbf{d}\sim\mathcal{CN}(0,\mathbf{R_d})$  and $\mathbf{h}\sim\mathcal{CN}(0,\mathbf{R_h})$  where $\mathbf{R_d}$ and $\mathbf{R_h}$ are the covariance matrices of $\mathbf{d}$ and $\mathbf{h}$, respectively. The covariance matrices are constructed such that their $(i,j)$-th element is given by $\mathbf{R}_{\mathbf{d}, {ij}}={\rm \sinc}\left({2 d_{ij}}/{\lambda}\right)$ and $\mathbf{R}_{\mathbf{h}, {ij}}={\rm \sinc}\left({2 r_{ij}}/{\lambda}\right)$ 
where  $\lambda$ is the operational wavelength, and $d_{ij}$ and $r_{ij}$ are the distances between $i$-th and $j$-th antennas of SU and RIS, respectively. It is worth noting that this covariance model captures the spatial correlation for UPA more accurately than a Kronecker-product-based model. For more details, please refer to \cite{Emil_CovMat}. 

The spectrum sensing problem can be formulated as a binary hypothesis testing problem between hypothesis $\mathcal{H}_0$ denotes the null hypothesis and hypothesis $\mathcal{H}_1$ denotes the presence of the PU. Thus, the received signal under these hypotheses can be expressed as
\begin{equation}
    \mathbf{y}_{k} = 
\begin{cases}
    \mathbf{w}_{k}, & \text{Under}\;\mathcal{H}_{0}\\
    \sqrt{\beta}\mathbf{d}_k \text{s}_k + \sqrt{\nu}\mathbf{H} \mathbf{\Phi} \mathbf{h}_k \text{s}_k + \mathbf{w}_{k}, & \text{Under}\;\mathcal{H}_{1}
    \label{eq_rec_sig}
\end{cases}
\end{equation}
for $k = 1,\ldots,K$, where $\text{s}_k$ corresponds to the $k$-th symbol transmitted by the PU, $\mathbf{w}_k$ is complex correlated Gaussian noise $\mathbf{w}_k\sim{\mathcal{CN}(0,\mathbf{R_w})}$ and $\mathbf{R_w}$ is the noise covariance matrix. 
In \eqref{eq_rec_sig}, $\beta = d_{o}^{-\xi}$ and $\nu = (d_{1}d_{2})^{-\xi}$ represent the path-losses along the direct and indirect links, respectively, where $d_o$, $d_1$ and $d_2$ represent the distances of the PU-SU, PU-RIS and RIS-SU links, respectively, and $\xi$ is the path-loss exponent.
The phase shift matrix of the passive RIS is denoted by $\bm{\Phi} = \text{diag}\left(\bm{\psi}\right) $ such that $\bm{\psi}^H = \left[e^{j\vartheta_{1}}, \dots , e^{j\vartheta_{N}} \right]^T$ and $\vartheta_{\rm N}$ is the phase shift provided by the $n$-th RIS element.  We assume that the PU transmission power is $P_s$, i.e., $\mathbb{E}\left[\text{s}_k \text{s}_k^{H}\right] = P_s$.
\section{Sensing using Maximum Eigenvalue Detection}
MED employs knowledge of the principal eigenvalues of sample covariance matrices of observed signals for hypothesis testing. This method exhibits superior performance compared to energy-based detection techniques, particularly in scenarios where the received signal is correlated \cite{Zeng2008MED}, as is the case considered in this paper. Thus, we employ the MED approach for spectrum sensing, assuming that both the multi-path fading and noise are correlated. The sample covariance matrix associated with the received signal can be defined as
\begin{equation}
    \hat{\mathbf{R}}_{\mathbf{y}} \triangleq \frac{1}{K}\sum\nolimits_{k=0}^{K-1} \mathbf{y}_{k} \mathbf{y}_{k}^{\textit{H}}.\label{eq:refer1}
\end{equation}
Applying MED for the detection problem given \eqref{eq_rec_sig}, the test statistic is given as
\begin{equation}
    {\rm TS} \triangleq \alpha_{\rm max},\label{eq:TS}
\end{equation}
   where $\alpha_{\max}$ is the largest eigenvalue of $\hat{\mathbf{R}}_{\mathbf{y}}$. The decision rule is
\begin{align}
    {\rm TS}\mathop{\lessgtr}_{\mathcal{H}_1}^{\mathcal{H}_0}\eta,
\end{align}
where $\eta$ is the decision threshold.
Thus, the detection and false alarm probabilities can be defined as
\begin{align*}
    \rm P_{D}(\mathbf{\eta}) &\triangleq \rm Pr(TS > \eta\, |\, \mathcal{H}_1),\\
    \text{~and~}\rm P_{FA}(\mathbf{\eta}) &\triangleq \rm Pr(TS > \eta\, |\, \mathcal{H}_0),
\end{align*}
respectively.
For the detection problem, the goal is to maximize the probability of detection $\rm P_D$ while minimizing the probability of false alarm $\rm P_{FA}$. However, maximizing $\rm P_D$ comes at the cost of raising $\rm P_{FA}$. This results in a trade-off between $\rm P_D$ and $\rm P_{FA}$ that needs to be carefully handled by appropriately determining the optimal choice of $\eta$. We will consider Neyman-Pearson's criteria \cite{kay2009fundamentals} for finding the decision threshold that maximizes the $\rm P_D$ for a given ${\rm P_{FA}}$. For this, it is necessary to characterize $\rm P_D$ and ${\rm P_{FA}}$ analytically. Determining the statistical distribution of the test statistics under both hypotheses is paramount.
For the considered system, the distributions of the received signals under both hypotheses given in \eqref{eq_rec_sig} can be determined as
\begin{align}
    \mathbf{y}_{k} &\sim 
    \begin{cases}
      \mathcal{CN}(0,\mathbf{R_w}), & \text{Under}\;\mathcal{H}_{0}, \\
      {\mathcal{CN}(0,\mathbf{R_s(\mathbf{\Phi})}+\mathbf{R_w})}, & \text{Under}\;\mathcal{H}_{1},
    \end{cases} \label{y_distribution}
\end{align}
where
\begin{equation}
    \mathbf{R_s}(\mathbf{\Phi}) = \beta P_{s}\mathbf{R_d}+\nu P_{s}\textbf{H} \bm{\Phi} \mathbf{R_h} \bm{\Phi}^H \textbf{H}^H,\label{eq:Rs}
\end{equation} 
is the covariance of the signal component for a given RIS phase shift configuration matrix $\mathbf{\Phi}$. Leveraging the benefits of RIS, it makes sense to configure it to maximize the mean value of the largest eigenvalue of the sample covariance matrix $ \hat{\mathbf{R}}_{\mathbf{y}}$ under $\mathcal{H}_1$ so that the detection probability can be improved, i.e.,
\begin{subequations}
\begin{align}
    \max_{\mathbf{\bm{\psi}}} \; &\mathbb{E}\left[\alpha_{\max}|\mathcal{H}_1\right], \label{alphamax} \\
    \text{such that}\, &\left|\mathbf{\bm{\psi}}_{i}\right| = 1, ~\text{for}~i=1,\dots,N,\label{eq:constraint}
\end{align}\label{eq:lambdamax_maximization}
\end{subequations}
where \eqref{eq:constraint} represents the unit modulus constraint of passive RIS. However, this constraint makes the problem non-convex; thus, problem \eqref{eq:lambdamax_maximization} becomes difficult to solve directly. Since $\hat{\mathbf{R}}_{\mathbf{y}}$ is a positive semi-definite,  we have $\alpha_{\max} < \text{Tr}(\hat{\mathbf{R}}_{\mathbf{y}})$, where $\text{Tr}(\cdot)$ represents trace operator. 
Thus, the mean values of $\alpha_{\max}$ under $\mathcal{H}_1$ can be bounded as
\[
\mathbb{E}\left[\alpha_{\max}|\mathcal{H}_1\right]\leq\mathbb{E}[\text{Tr}(\hat{\mathbf{R_y}})|\mathcal{H}_1]=\text{Tr}(\mathbf{R_s}(\mathbf{\Phi}))+\text{Tr}(\mathbf{R_w}).
\]
Using this identity, we aim to select $\mathbf{\Phi}$  that maximizes the trace of $\mathbf{R_s}$, i.e. maximizing the upper bound on $\mathbb{E}[\alpha_{\max}]$. 
Using \cite[Appendix A]{Taricco} and $\mathbf{\Phi}={\rm diag}(\bm{\psi})$, we can write
\begin{align*}
\text{Tr}(\mathbf{R_s(\mathbf{\Phi})})&= \beta P_{s}\text{Tr}(\mathbf{R_d}) + \nu P_{s}\text{Tr}(\mathbf{H\Phi R_h}\mathbf{\Phi}^H \mathbf{H}^H), \nonumber\\
&= \beta P_{s}\text{Tr}(\mathbf{R_d}) + \nu P_{s}\bm{\psi}^H \left((\mathbf{H}^H \mathbf{H})\odot \mathbf{R_h}\right)\bm{\psi},
\end{align*}
where $\odot$ is the Hadamard product.
Thus, the trace maximization problem can be written as
\begingroup
\setlength\abovedisplayskip{0.35pt}
\setlength\belowdisplayskip{0.35pt}
\begin{subequations}
\begin{align}
        \max_{\mathbf{\bm{\psi}}}\; &\bm{\psi}^H \left((\mathbf{H}^H \mathbf{H})\odot \mathbf{R_h}\right)\bm{\psi}, \\
        \text{s.t.}\, &\left|\mathbf{\bm{\psi}}_{i}\right| = 1.
\end{align}
\end{subequations}
\endgroup
Further, using \cite[Theorem 4]{kota2023statistically}, we can directly obtain the optimal RIS phase shift configuration as
\begingroup
\setlength\abovedisplayskip{0.1pt}
\setlength\belowdisplayskip{0.1pt}
\begin{equation}
    \bm{\psi}^\star = e^{j\vartheta}\mathbf{1}_{N}, \label{eq:optpsi}
\end{equation}
\endgroup
where $\mathbf{1}_{N}$ a $N\times1$ vector with unit entries and $\vartheta\in[0, 2\pi]$. As $\bm{\psi}^\star$ maximizes the upper bound on the mean value of $\alpha_{\max}$, it naturally becomes a sub-optimal solution for the original problem given in \eqref{eq:lambdamax_maximization}.
It is worth noting that the solution given in \eqref{eq:optpsi} provides the same phase shift to all the signal components reflecting off the surfaces of RIS elements. This reduces the operationality of  RIS to intelligent mirrors.
Thus, for the optimally configured RIS phase-shift matrix, the covariance matrix under $\mathcal{H}_1$ become
\begingroup
\setlength\abovedisplayskip{0.1pt}
\setlength\belowdisplayskip{0.1pt}
\begin{equation}
    \mathbf{R_s}(\mathbf{\Phi}^\star) = \beta P_{s}\mathbf{R_d} + \nu P_{s}\mathbf{H R_h}\mathbf{H}^H. 
\end{equation}
\endgroup
Note that the sample covariance matrix, as defined in \eqref{eq:refer1}, is random, which makes the resulting largest eigenvalue a random variable. Thus, to obtain $\rm P_D$ and $\rm P_{FA}$, it is crucial to first derive the distribution of the largest eigenvalue under both $\mathcal{H}_{0}$ and $\mathcal{H}_{1}$. From \eqref{eq:refer1} and \eqref{y_distribution}, it is clear that $\hat{\mathbf{R}}_{\mathbf{y}}$ is a central Wishart matrix. This promotes the application of random matrix theory \cite{ge2021large} to derive the distribution of the largest eigenvalue of the central Wishart matrix. On this line, huge research efforts are devoted to the literature to characterize this distribution. However, the available results mostly focus on 1) the approximate distribution of the largest eigenvalue in the asymptotic region and 2) the Wishart matrix with uncorrelated entries. For instance, \cite{ge2021large}, characterize the distribution using Tracy-Widom in the asymptotic region. An exact largest eigenvalue distribution is studied in \cite{Khatri} and then thoroughly extended in \cite{Chiani_2014} to obtain the close form results for the iid noise case. However, these results do not apply to the considered system setting as it includes a general case of correlated fading and correlated noise with a moderate number of receiving antennas.
The exact distribution of the largest eigenvalue of a central Wishart matrix is given in the following theorem
\vspace*{-1.4ex}
\begin{theorem}\label{thm:Theorem1}
    Let $\mathbf{X}$ be a $p$ $\times$ $q$ matrix whose columns are $p$-variate zero-mean complex Gaussian distributed, with covariance matrix $\mathbf{\Sigma}$, i.e., $\mathbf{X} \sim \mathcal{CN}(0, \mathbf{\Sigma})$. Say $n = \max(p,q)$ and $m = \min(p,q)$. Let $\mathbf{A} \triangleq \mathbf{XX^H}$ be the central Wishart matrix and let $\alpha_{\rm 1}, \alpha_{\rm 2}, \dots , \alpha_{\rm m}$ be the ordered eigenvalues of matrix $\mathbf{A}$. Let $\lambda_{1},\lambda_{2},\ldots,\lambda_{m}$ be the ordered eigenvalues of $\mathbf{\Sigma^{-1}}$, then the CDF of the largest eigenvalue $\alpha_m$ of $\mathbf{A}$ is given by
    \begingroup
    \setlength\abovedisplayskip{0.35pt}
    \setlength\belowdisplayskip{0.35pt}
        \begin{align}
            {\rm Pr}(\alpha_m \leq \eta) &= \frac{c}{|\mathbf{V}|} {\rm det} \left[\bm{\Lambda}(\eta,\lambda)\right],\label{eq:cdf}
        \end{align}
    \endgroup
    where $\gamma(\cdot\,,\cdot)$ is lower incomplete gamma function and $(i,j)$-th element of $\bm{\Lambda}(\eta,\lambda)$ is
    \begin{align}
        \{\bm{\Lambda}(\eta,\lambda)\}_{ij} &= \frac{1}{\lambda^{n-i+1}}\gamma(n-i+1,\eta\lambda_j),\label{thm10}\\
       \text{and~} |\mathbf{V}| &= \prod\nolimits_{i,j=1,\dots,m,i<j}^{m}(\lambda_{i}-\lambda_{j}).\label{thm10a}
        \end{align}
\end{theorem}
\begin{proof}
    Please refer to Appendix \ref{app:Thm1_proof} for the proof.
\end{proof}
\vspace*{-1.2ex}
Using Theorem \ref{thm:Theorem1}, we can determine the largest eigenvalue distributions of sample covariance matrices under $\mathcal{H}_0$ and $\mathcal{H}_1$, and accordingly use them to evaluate the false alarm and detection probabilities, as done in the following subsections.
\vspace*{-1.3ex}
\subsection{False alarm probability}\label{Pfa}
\vspace*{-1.0ex}
Under the null hypothesis $\mathcal{H}_0$, where the signal is not transmitted, the received signal solely comprises noise. Thus, the observed sample covariance matrix becomes
\begingroup
\setlength\abovedisplayskip{1.7ex}
\setlength\belowdisplayskip{1.7ex}
   \begin{equation}
       \hat{\mathbf{R}}_{\mathbf{y}} = \frac{1}{ K}\sum\nolimits_{ k=0}^{ K-1} \textbf{w}_{k} \textbf{w}_{k}^{\rm H}, \label{eq:noiseSC}
   \end{equation}
\endgroup
   and the covariance matrix of $\mathbf{w}_k$ is $\mathbf{R_w}$. Hence, the test statistics, as defined in \eqref{eq:TS}, becomes the largest eigenvalue of \eqref{eq:noiseSC}. Let $0 < \lambda_{1}^{\mathbf{w}} < \lambda_{2}^{\mathbf{w}} < \dots < \lambda_{m}^{\mathbf{w}}$ represent the ordered eigenvalues of $\mathbf{R}^{-1}_{\mathbf{w}}$. Using Theorem \ref{thm:Theorem1} and substituting $\bm{\Sigma}^{-1} = \mathbf{R}^{-1}_{\mathbf{w}}$, we can obtain the exact distribution of TS and  thereby get a closed-form expression for $\rm P_{FA}$ as
\begin{align}
            \rm P_{FA}(\eta) & = 1 - \frac{c}{|\mathbf{V}_{\lambda^{\mathbf{w}}}|} {\rm det} \left[\bm{\Lambda}(\eta,\lambda^{\mathbf{w}}_{j})\right],\label{eq:PFADist}
\end{align} 
where $\bm{\Lambda}(\eta,\lambda^{\mathbf{w}})$ is given in \eqref{thm10} and $|\mathbf{V}_{\lambda^\mathbf{w}}|$ can be obtained using  \eqref{thm10a} with $\lambda_i=\lambda_i^\mathbf{w}$. 
    To achieve a desired $\rm P_{FA}$, the detection threshold $\eta$ can be obtained by inverting \eqref{eq:PFADist}. We adopt a numerical approach to invert \eqref{eq:PFADist} for obtaining $\eta$.
\vspace*{-1.3ex}
\subsection{Detection Probability}
\vspace*{-1.3ex}
When a signal is present, i.e., hypothesis $\mathcal{H}_1$, the received signal consists of both the PU signal and noise components. The sample covariance matrix for this case is given as
\begingroup
\setlength\abovedisplayskip{1.7ex}
\setlength\belowdisplayskip{1.7ex}
\begin{equation}
\hat{\mathbf{R}}_{\mathbf{y}} = \frac{1}{K}\sum_{k=0}^{K-1} \beta\mathbf{d}_k\mathbf{d}_k^H + \nu\mathbf{H} \mathbf{\Phi} \mathbf{h}_k \mathbf{h}_k^H \mathbf{\Phi}^H \mathbf{H}^H  + \mathbf{w}_k\mathbf{w}^{H}_k, \label{eq:signalSC}
\end{equation}
\endgroup
where the covariance matrix of the received signal under $\mathcal{H}_1$ is $\mathbf{R_s(\bm{\Phi}^\star) + R_w}$, as defined in \eqref{y_distribution}. The test statistics is the maximum eigenvalue of \eqref{eq:signalSC}. It is to be noted that the above sample covariance matrix is observed with the optimally configured RIS phase shifts, which aim to maximize the expected value of test statistics. Now, using  Theorem \ref{thm:Theorem1}, we can  evaluate the distribution of  test statistics with $\bm{\Sigma}^{-1} = \left(\mathbf{R_s(\bm{\Phi}^\star) + R_w}\right)^{-1}$. Let $0 < \lambda_{1}^{\mathbf{y}} < \lambda_{2}^{\mathbf{y}} < \dots < \lambda_{m}^{\mathbf{y}}$ be the ordered eigenvalues of $\left(\mathbf{R_s(\bm{\Phi}^\star) + R_w}\right)^{-1}$. Hence, the distribution of the largest eigenvalue is given as
\begin{equation}
    \begin{aligned}
        \rm P_D(\eta) &= 1 - \frac{c}{|\mathbf{V}_{\lambda^{\mathbf{y}}}|} {\rm det} \left[\bm{\Lambda}(\eta,\lambda^{\mathbf{y}}_{j})\right],
    \end{aligned}
\end{equation}
where $\bm{\Lambda}(\eta,\lambda^{\mathbf{y}})$ is given in \eqref{thm10} and $|\mathbf{V}_{\lambda^\mathbf{y}}|$ can be obtained using  \eqref{thm10a} with $\lambda_i=\lambda_i^\mathbf{y}$.

\vspace*{-0.3ex}
 \section{Numerical Results}
 This section presents a numerical analysis of the proposed MED for RIS-aided spectrum sensing. In particular, we will verify the false alarm and detection probabilities via simulations. Next, we will discuss how the proposed MED's ROC curve performs in terms of the various system design parameters. For numerical analysis, we consider the number of received antenna $M = 8$, the number of RIS element $N = 32$, the distance between subsequent RIS element $d = \lambda/2$,  the number of observations $K = 10$, the transmission power of PU $P_s = 10\;$dBm,  direct link distance $d_o = 30$ m and path-loss exponent $\xi = 3$, unless we mention otherwise.

 Recall that RIS phase shift is configured according to \eqref{eq:optpsi} to maximize the mean of the observed largest eigenvalue under $\mathcal{H}_1$. The signal-to-noise ratio (SNR) of the received signal is 
\begin{align*}
    {\rm SNR}=\Upsilon\|\sqrt{\mu}\mathbf{H}\mathbf{\Phi}^\star\mathbf{h}+\mathbf{d}\|^2,
\end{align*}
where $\Upsilon=\frac{\beta P_S}{\sigma_W^2}$ is the mean SNR received per antenna on the direct link and $\mu=\frac{\nu}{\beta}$ represents the ratio of path-losses of the indirect and direct links. Without the loss of generality, we will consider that $d_o = \kappa d_1d_2$ where $\kappa \in [0,1]$ and thus model the path-loss ratio as $\mu = \kappa^{\xi}$ to capture the effect of the relative difference in path-losses observed by the direct and indirect links.  Further, we model the noise covariance matrix as 
$\mathbf{R_{w,}}_{ij} = \rho^{\left|i-j\right|}\sigma_{W}^{2}$ such that $\sigma_{W}^2$ is the noise variance and $\rho$ is the correlation coefficient of $i$-th and $j$-th noise components. We consider $\kappa=\frac{1}{3}$ and the correlation factor $\rho = 0.2$.
\begingroup
\setlength{\belowcaptionskip}{-10pt}
\begin{figure}[t!]
  \centering\vspace{-3mm}
    \includegraphics[width=0.35\textwidth]{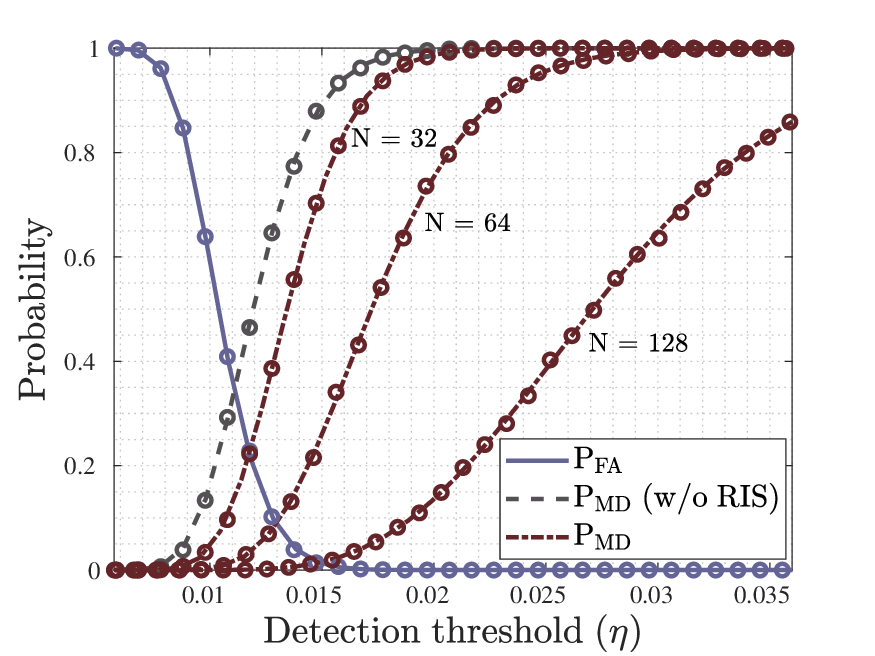}\vspace{-3mm}
  \caption{Probabilities of false alarm $\rm P_{FA}$ and missed detection $\rm P_{MD} = 1-P_D$ vs. threshold $\eta$. The lines indicate the analytical curves, whereas the circular markers represent the simulated results.}
  \vspace{-3.3ex}
  \label{fig:P_vs_threshold}
\end{figure}
\endgroup
Fig. \ref{fig:P_vs_threshold} verifies the derived false alarm and miss detection probabilities using simulation results for $\Upsilon=-10\, \rm{dB}$. The simulation results (denoted by markers) match exactly with the analytical results derived for the ${\rm P_{FA}}$ and ${\rm P_{MD}}$ (denoted by lines). As expected, these probabilities exhibit the performance trade-off w.r.t. to the decision threshold $\eta$. The figure also shows that the missed detection probability ${\rm P_{MD}}$ performs better when RIS is utilized for the sensing. It can be seen that the  ${\rm P_{MD}}$ reduces significantly with the increase in RIS elements $N$ for a given $\eta$. Basically, by using a large dimensional RIS, it is possible to increase the {\em deflection coefficient} \cite{kay2009fundamentals} corresponding to the test statistics distributions observed under $\mathcal{H}_0$ and $\mathcal{H}_1$, which essentially will help to improve the performance of ROC curves, as it will be highlighted shortly.

Next, in Fig. \ref{fig:ROC_N_128},  we show the ROC, i.e., $\rm P_D$ vs $\rm P_{FA}$ at $\Upsilon = -10,\,-8,\,\text{and}\,-5$ dB.  The figure shows that the performance improves with an increase in $\Upsilon$, which is quite expected. However, it can be seen that the MED performance without RIS for $\Upsilon = -10$ dB is very poor. In contrast, including RIS with $N = 32$ gives a significantly higher detection probability for the same $\Upsilon$. This ensures that RIS can be useful to improve the spectrum sensing performance without necessarily increasing the transmission power or using more sophisticated/complex designs for the receiver. 
Next, Fig. \ref{fig:ROCvaryingN} shows that the ROC curve of the proposed MED improves with the increase in the number of RIS elements. In particular, it can be observed that the detection probability rises sharply with just doubling the number of RIS elements. Therefore, it is possible to achieve the detection probability close to $1$ for a very small probability of false alarm using moderately large dimensional RIS.\\

\begin{figure}[t]
  \centering\vspace{-3mm}
\includegraphics[width=0.35\textwidth]{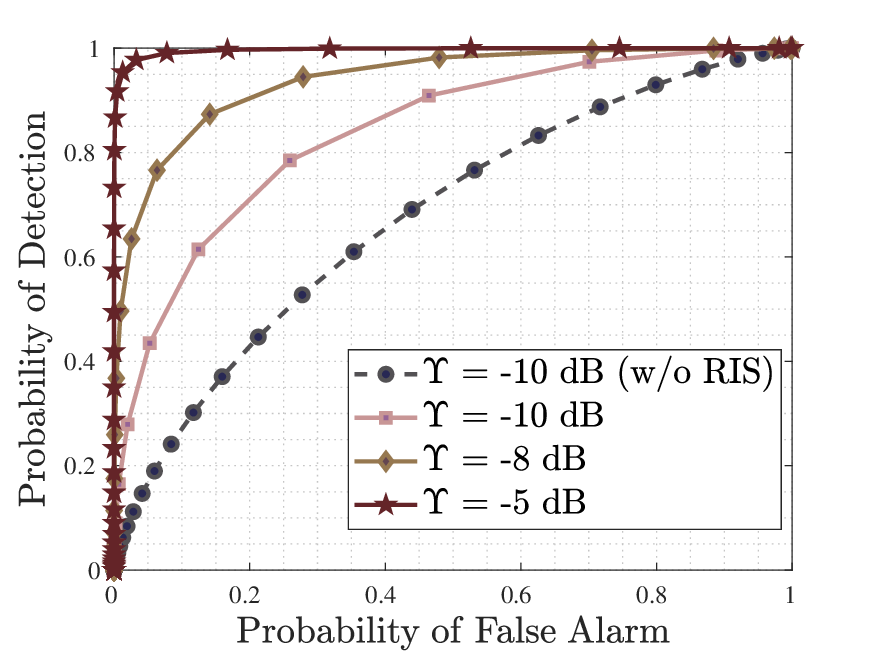}\vspace{-3mm}
  \caption{ROC: $\rm{P_D}$ vs $\rm{P_{FA}}$ for various SNR $\Upsilon$.}
  \label{fig:ROC_N_128}
\end{figure}
\vspace{-1.3ex}

\begin{figure}[t]
  \centering\vspace{-3mm}
\includegraphics[width=0.35\textwidth]{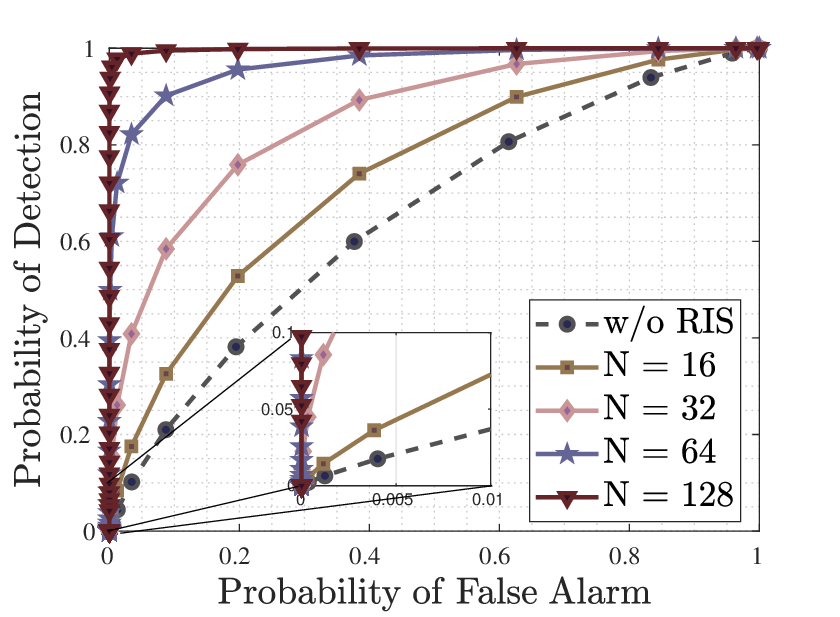}\vspace{-3mm}
  \caption{ROC: $\rm{P_D}$ vs $\rm{P_{FA}}$ for different number of RIS elements $N$.}
  \vspace{-3.3ex}
  \label{fig:ROCvaryingN}
\end{figure}\vspace{-3mm}



Fig. \ref{fig:OptPhivaryingN_OptPhivaryingsnr} shows the benefit of configuring the RIS phase shifts using the simple proposed sub-optimal solution given in \eqref{eq:optpsi} when the spacing between RIS elements is $d = \lambda/2$. Fig. \ref{fig:OptPhivaryingN_OptPhivaryingsnr} (Left) indicates the ROC performance for various values of RIS elements $N$ when $\Upsilon = -10$ dB. The figure shows that the ROC improves significantly when sub-optimally configured RIS is used, especially when $N$ is large. This performance improvement is attributed to the fact of utilizing properly configured RIS phase shifts that maximizes the trace of covariance matrix $\mathbf{R_s}$.
This also verifies that the correlated fading  aids the detection performance. 
Fig. \ref{fig:OptPhivaryingN_OptPhivaryingsnr} (Right) shows that the detection probability ${\rm P_D}$ (for ${\rm P_{FA}} = 0.1$) improves with the increase in $N$ for various values of $\Upsilon$. It can be observed that the ${\rm P_D}$ approaches $1$ for smaller values of $N$ when the operating $\Upsilon$ is larger and vice-versa. In addition, the figure also verifies the enhanced performance of ${\rm P_D}$ with sub-optimally configured RIS. However, the improvement diminishes with the increase in $\Upsilon$. This is due to the overall SNR along the direct link is sufficiently high and thus it reduces necessity of  RIS .  

\begin{figure}[t]
  \centering
\hspace{-1mm}\includegraphics[width=0.25\textwidth]{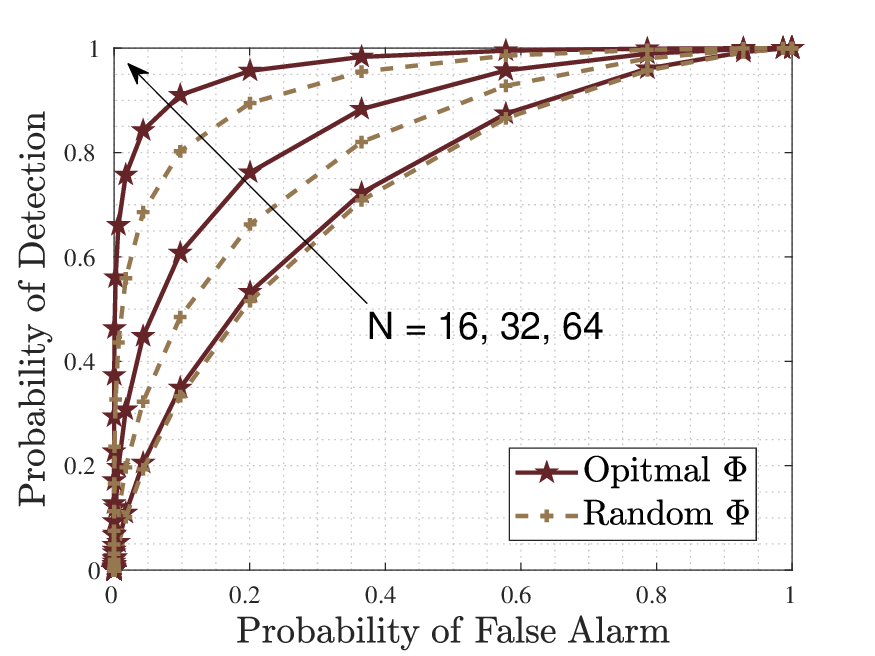}
\hspace{-4mm}\includegraphics[width=0.25\textwidth]{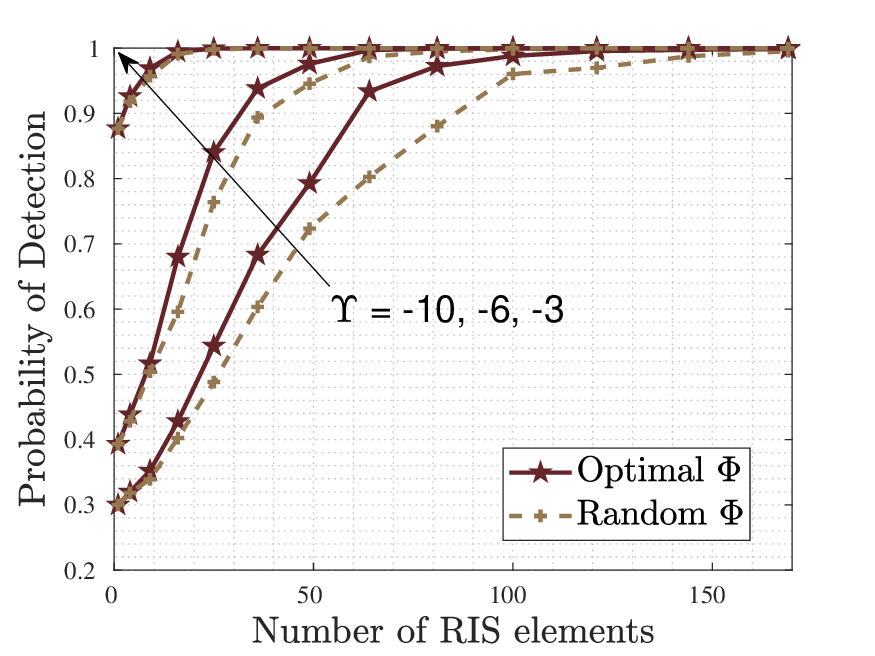}\vspace{-3mm}
  \caption{Left: ROC for different $N$. Right: $\rm{P_D}$ vs. $N$ for different $\Upsilon$.}
  \vspace{-3.3ex}
  \label{fig:OptPhivaryingN_OptPhivaryingsnr}
\end{figure}

\section{Conclusion}
In this paper, we investigate the performance of MED for RIS-aided spectrum sensing. We consider that the PU-RIS and PU-SU link experience correlated Rayleigh fading, and RIS is placed such that it has a strong LoS link with SU.
For this setup, we derived an exact distribution of the largest eigenvalue of the sample covariance matrix, which we next utilized to derive the detection and false alarm probabilities. In addition, we also configure the RIS phase shift matrix optimally to maximize the mean of the test statistics under the signal present hypothesis so that the detection probability can be improved.
It is worth noting that this paper is the first to present an exact analytical characterization of MED for RIS-aided spectrum sensing when both the multipath fading and noise are correlated.
Our numerical analysis demonstrates that the ROC curve of the MED exhibits improvement with the increase in the number of RIS elements and mean received SNR. It also shows that the detection probability approaches unity for many RIS elements, especially when the optimally configured RIS phase shifts are used for low SNR.  

\vspace{-1.3ex}
\appendix 
\label{app:Thm1_proof} 
 Let $\alpha_1, \alpha_2,\ldots, \alpha_m$ be eigenvalues of $\mathbf{A}$, and the joint density function of ($\alpha_1, \alpha_2,\ldots, \alpha_m$) is given in \cite{James1964}
    \begin{equation}
    \begin{aligned}
      f(\alpha_1, \alpha_2,\ldots, \alpha_m) &= \left|\mathbf{\Sigma}\right|^{-n}  \mathcal{_{\text{0}}\widetilde{F}_{\text{0}}}(-\mathbf{\Sigma}^{-1}, \mathbf{A}) \frac{\pi^{m(m-1)}}{\widetilde{\Gamma}_m(n)\widetilde{\Gamma}_m(m)} \\ 
      & \times |\mathbf{A}|^{n-m} \prod\nolimits_{i,j=1,i<j}^{m}(\alpha_i - \alpha_j)^2, \label{thm1}
    \end{aligned}
    \end{equation}
    where $n = \text{max}\left(p,q\right)$ and $m = \text{min}\left(p,q\right)$, $\mathcal{_{\text{0}}\widetilde{F}_{\text{0}}}(\cdot\,,\cdot)$ is the hypergeometric function with matrix argument,
    \[
        \widetilde{\Gamma}_s(t) = \pi^{s(s-1)/2} \prod\nolimits_{i=1}^{s}\Gamma(t-i+1),
    \]
    and $\Gamma(\cdot)$ is gamma function. Now, from \cite{GROSS1989}, we have the following identity
    \begin{equation}
        \mathcal{_{\text{0}}\widetilde{F}_{\text{0}}}(-\mathbf{\Sigma}^{-1}, \mathbf{A}) = \beta_m \frac{\left|\bm{\Lambda}(\alpha,\lambda)\right|}{\prod_{i,j = 1,i<j}^{m}(\lambda_i - \lambda_j)(\alpha_i - \alpha_j)}, \label{thm2}
    \end{equation}
    where $\{\bm{\Lambda}(\alpha,\lambda)\}_{ij}=$ $_{\text{0}}F_{\text{0}}(-\lambda_j\alpha_i)$, $_{\text{0}}F_{\text{0}}(\cdot)$ is hypergeometric function with scalar input, $\lambda_j$ is the $j$-th eigenvalue of $\Sigma^{-1}$ and
    \[
     \beta_m = \prod\nolimits_{a=1}^{m}(a-1)!.
    \]
    Now, substituting \eqref{thm2} in \eqref{thm1} gives us,
     \begin{equation}
        \begin{aligned}
           f(\alpha_1, \alpha_2,\ldots, \alpha_m) &= \frac{c}{|\mathbf{V}|} |\bm{\Lambda}(\alpha,\lambda)| |\mathbf{W}| \prod\nolimits_{i=1}^{m}\alpha_i^{n-m} \label{thm3}
        \end{aligned}
    \end{equation}
    \begin{align*}
        \text{where~} c &= \frac{|\mathbf{\Sigma}|^{-n}\beta_m}{\prod_{i=1}^{m}\Gamma(n-i+1)\Gamma(m-i+1)},\\
        \text{and~} |\mathbf{V}| &= \prod\nolimits_{i,j=1,\dots,m,i<j}^{m}(\lambda_{i}-\lambda_{j}).
    \end{align*}
    Note that $\mathbf{V}$ and $\mathbf{W}$ are Vandermonde matrices such that
    \begin{equation}
        \mathbf{V} = \left[\begin{matrix}
                                \lambda_{1}^{m-1} & \lambda_{1}^{m-2} & \cdots & 1 \\
                                \lambda_{2}^{m-1} & \lambda_{2}^{m-2} & \cdots & 1 \\
                                \vdots & \vdots & \ddots & \vdots \\
                                \lambda_{m}^{m-1} & \lambda_{m}^{m-2} & \cdots & 1
                            \end{matrix}\right], \label{thm4}
    \end{equation}
and $\mathbf{W}$ is defined in similar way as \eqref{thm4}, with $\alpha_i$'s instead of $\lambda_i$'s.
To find the CDF of the largest eigenvalue $\alpha_m$, we will have to integrate \eqref{thm3} over region $D = \left[0<\alpha_{1}<\ldots<\alpha_{m}<\eta\right]$. Hence, the CDF will be given by
    \begin{align*}
           \text{Pr}&(\alpha_m \leq \eta) = \frac{c}{|\mathbf{V}|} \int_{D} \left|\bm{\Lambda}(\alpha,\lambda)\right| \left|\mathbf{W}\right|\prod_{i=1}^{m}\alpha_i^{n-m} \prod_{i=1}^{m}d\alpha_i\numberthis\label{thm5}
    \end{align*}
    Let $I = \left|\bm{\Lambda}(\alpha,\lambda)\right|\left|\mathbf{W}\right| \prod_{i=1}^{m}\alpha_i^{n-m}$ be the integrand function in \eqref{thm5}.
    The above product of two determinants can be written as
    \begin{align*}
        &\left|\bm{\Lambda}(\alpha,\lambda)\right| \left|\mathbf{W}\right| = \sum_{(\sigma_1, \dots, \sigma_m)} \sum_{(i_1,\dots, i_m)} \operatorname{sign}(i_1, \dots, i_m) \times\\
            &   \alpha_{\sigma_1}^{m-1} \, _{\text{0}}F_{\text{0}}(-\lambda_{i_1}\alpha_{\sigma_1}) 
             \alpha_{\sigma_1}^{m-2} \, _{\text{0}}F_{\text{0}}(-\lambda_{i_2}\alpha_{\sigma_2}) \dots  _{\text{0}}F_{\text{0}}(-\lambda_{i_m}\alpha_{\sigma_m}). 
    \end{align*}
     Now, using this, we can simplify $I$  as
    \begin{equation*}
        I = |\bm{\Lambda^{'}}(\alpha,\lambda)|,  
    \end{equation*}
    where $\{\bm{\Lambda'}(\alpha,\lambda)\}_{ij} = \alpha_{i}^{n-i} {_{\text{0}}F_{\text{0}}}(-\lambda_{j}\alpha_{i}) $.
    Further by substituting $I$ in \eqref{thm5} we get,
    \begin{equation}
        \text{Pr}(\alpha_m \leq \eta) = \frac{c}{|\mathbf{V}|} \int_{D} \left|\bm{\Lambda'}(\alpha,\lambda) \right|\prod_{i=1}^{m}d\alpha_i.\nonumber
    \end{equation}
    Next, by applying \cite[Lemma 1]{Khatri1969LEMMA}, we can write
    \begin{equation}
        \text{Pr}(\alpha_m \leq \eta) = \frac{c}{|\mathbf{V}|} \left|\int_{0}^{\eta}\bm{\Lambda}(y,\lambda) dy\right|, \label{thm15}
    \end{equation}
    where 
        $\{\bm{\Lambda}(y,\lambda)\}_{ij} = y^{n-i} {_{\text{0}}F_{\text{0}}(-\lambda_j\alpha_i)}$.
    Note that the $_{\text{0}}F_{\text{0}}$ is a special case hypergeometric function which can be written as
    \begin{equation*}
        _{\text{0}}F_{\text{0}}(-\lambda_j\alpha_i) = e^{-y\lambda_j}.\label{thm9}
    \end{equation*}
    Finally, applying the above identity, we can simplify the integral in \eqref{thm15} as
    \begin{equation}
    \begin{aligned}
        \int_{0}^{\eta}\{\bm{\Lambda}(y,\lambda)\}_{ij} dy&= \int_{0}^{\eta} y^{n-i} e^{-y\lambda_j} dy,
    \end{aligned}
    \end{equation}
    which will give us \eqref{thm10} and further substituting it into (\ref{thm15}) will complete the proof.

\end{document}